\documentclass[journal]{IEEEtran}
\usepackage{hyperref,graphicx}

\graphicspath{{./figs/}{.}}
\usepackage{amsmath}
\usepackage{amssymb}
\usepackage{amsthm}
\usepackage{bm} 
\usepackage{dirtytalk} 
\usepackage{booktabs} 
\usepackage{multirow}
\usepackage{booktabs}
\usepackage{cite}
\usepackage[monochrome]{xcolor}
\usepackage{soul}
\sethlcolor{green}
\usepackage{lipsum}
\usepackage{stackrel}
\usepackage[lofdepth,lotdepth]{subfig}

\newcommand{\bmi}{\bm i}

\renewcommand{\epsilon}{\varepsilon}
\renewcommand{\phi}{\varphi}

\DeclareMathOperator{\var}{var}

\newtheorem{theorem}{Theorem}
\newtheorem{proposition}[theorem]{Proposition}
\newtheorem{Corollary}[theorem]{Corollary}
\newtheorem{lemma}[theorem]{Lemma}

\begin{document}

\title{Asymptotic regime for {\color{blue} impropriety} tests of complex random vectors}

\author{Florent Chatelain, Nicolas Le Bihan and Jonathan H. Manton}

\author{Florent Chatelain,~\IEEEmembership{Member,~IEEE,}
        Nicolas~Le Bihan~\IEEEmembership{Member,~IEEE,}
        and~Jonathan~H.~Manton,~\IEEEmembership{Fellow,~IEEE}
\thanks{F. Chatelain and N. Le Bihan are with Univ. Grenoble Alpes, CNRS,
Grenoble INP${}^\dag$, GIPSA-Lab, 38000 Grenoble, France.
{ ${}^\dag$: Institute of Engineering Univ. Grenoble Alpes}. Jonathan H.
Manton is with the Department of Electrical and Electronic Engineering, The
University of Melbourne, Victoria, 3010 Australia.}}%


\maketitle

\begin{abstract}
Impropriety testing for complex-valued vectors 
has been considered lately due to potential applications 
ranging from digital communications to {\color{blue}complex media imaging}. This paper
 provides new results for such tests in the asymptotic regime, {\em i.e.} when
the vector {\color{blue}dimension} and sample size grow commensurately to infinity. The studied tests
are based on invariant statistics named {\color{blue} \em impropriety coefficients}.
Limiting distributions for these statistics are derived, together with those of
the Generalized Likelihood Ratio Test (GLRT) and Roy's test, in the
Gaussian case. This characterization in the asymptotic regime allows also to
identify a phase transition in Roy's test with potential application in
detection of complex-valued low-rank {\color{blue} subspace} corrupted by proper noise in large
datasets. Simulations illustrate the accuracy of the proposed asymptotic
approximations.
\end{abstract}

\section{Introduction}
\label{sec:intro}
Testing for properness consists of deciding whether $N$-dimensional complex
random vector ${\bf z} \in {\mathbb C}^N$ is {\em proper} or {\em improper}.
{\color{blue}Let us recall that a complex vector
${\bf z}$ is called {\em circular}  when
${\bf z}$ is equal in distribution to $e^{\bmi \theta}{\bf z}$ ({\em i.e.} the distribution of
${\bf z}$ is invariant by a rotation of angle $\theta$ in the complex domain
\cite{Comon:1994,Picinbono:1994}). {Circularity} for Gaussian complex vectors is actually called {\em properness}}. {\color{blue} The problem of testing weither a complex Gaussian vector is {\em proper} or not has several potential applications in signal processing and} has been considered by several authors, using different means, including Generalized Likelihood Ratio Tests
(GLRT) \cite{Ollila:2004,Schreier:2006}, locally most powerful (LMP) test
\cite[Chapter 3]{Schreier:2010} or frequency domain tests \cite{Chandna:2017}.
The asymptotic behavior of GLRT was studied for large sample sizes in the case of
random variables ($N=1$) \cite{Delmas:2011} or small/fixed values of $N$
\cite{Walden:2009}. The situation where both the dimension of the complex vector
and the size of the sample tend to infinity was not considered until our recent preliminary study \cite{Chatelain:2019}. This article formalizes and extends \cite{Chatelain:2019}, therefore filling {\color{blue} a} gap and provides insight
into the asymptotic behavior of {\em {\color{blue} impropriety}} test when the vector
dimension and sample size grow commensurately to infinity. {\color{blue} Practical applications of this asymptotic regime occur for example in communications when a large number of dense arrays are deployed, {\em i.e.} for massive MIMO \cite{zarei:2016} and cognitive radio \cite{axell:2012}, as well as in fMRI \cite{adali:2011} or phase retrieval and imaging in complex media \cite{dong:2019}.}

Central to many of the tests available in the literature, the set of {\em
invariant parameters} was first considered in \cite{Andersson:1975}, allowing for the
derivation of {\em invariant statistics} used in \cite{Walden:2009}. Invariant
parameters are in one-to-one correspondence with {\em canonical correlation
coefficients} \cite{Schreier:2006} as explained in \cite{Schreier:2010}, {\color{blue}and referred
to as {\em impropriety coefficients} \cite{hellings2019}.}

In this article, we consider $N$-dimensional complex-valued centered
random vectors with Cartesian form such as ${\bf z}={\bf u} + \bmi {\bf v}$,
{\em i.e.} ${\bf u}$ and ${\bf v}$ are $N$-dimensional real vectors with zero
mean, {\em i.e.} ${\bf u} \in {\mathbb R}^N$, ${\bf v} \in {\mathbb R}^N$, and
${\mathbb E}[{\bf u}]={\mathbb E}[{\bf v}]={\bf 0}$. Two {\em augmented
representations} are classically used in the literature to study complex vectors
${\bf z} \in {\mathbb C}^N$, namely the {\em real augmented representation}
${\bf x}$ and the {\em complex augmented representation} $\tilde{\bf z}$. The
former consists of representing ${\bf z} \in {\mathbb C}^N$ by a twice larger
real-valued vector ${\bf x}= \left[{\bf u}^T , {\bf v}^T \right]^T \in {\mathbb
R}^{2N}$ made up from the real and imaginary parts of ${\bf z}$, while the latter
consists of using a twice larger complex-valued vector $\tilde{\bf z}=
\left[{\bf z}^T , {\bf z}^{*T} \right]^T \in {\mathbb C}^{2N}$ containing ${\bf
z}$ and its conjugate ${\bf z}^*$. Both representations are equivalent and
easily connected using linear mappings given for example in
\cite{Schreier:2010}. In this paper, we make use of the real representation
${\bf x} \in {\mathbb R}^{2N}$.

Recalling that ${\mathbb E}[{\bf z}]={\bf 0}$, and thus that ${\mathbb E}[{\bf
x}]={\bf 0}$, second order statistics of ${\bf z} \in {\mathbb C}^N$ are
contained in the real-valued covariance matrix ${\bf C} \in {\mathbb
R}^{2N\times 2N}$ of the real representation vector ${\bf x}$, which reads:
\begin{equation}
{\mathbb E}[{\bf x}{\bf x}^T]= {\bf C} = \begin{pmatrix}
{\bf C}_{{\bf u}{\bf u}} & {\bf C}_{{\bf u}{\bf v}} \\
{\bf C}_{{\bf v}{\bf u}} & {\bf C}_{{\bf v}{\bf v}}\\
\end{pmatrix}
\label{eq:covmat}
\end{equation}
where ${\bf C}_{{\bf a}{\bf b}} \in {\mathbb R}^{N \times N}$ denotes the real-valued (cross)covariance matrix between real vectors ${\bf a}$ and ${\bf b}$, and with ${\bf C}_{{\bf b}{\bf a}}= {\bf C}_{{\bf a}{\bf b}}^T$. A complex-valued Gaussian vector ${\bf z} \in {\mathbb C}^N$ is called {\em proper} {\em iff} the following two conditions hold:
\begin{equation}
\begin{array}{lcr}
  {\bf C}_{{\bf u}{\bf u}} = {\bf C}_{{\bf v}{\bf v}} & \text{and} & {\bf C}^T_{{\bf u}{\bf v}} = - {\bf C}_{{\bf u}{\bf v}}.
 \end{array}
\label{eq:prop}
\end{equation}
If these conditions are not fulfilled, then ${\bf z}$ is called {\em
improper}.
{\color{blue}
{\em Properness} thus means that real and imaginary parts, {\em
i.e.} ${\bf u}$ and ${\bf v}$, have the same covariance matrix and their
cross-covariance is skew-symmetric (2). When using the complex
representation, {\em properness} is equivalent to having ${\mathbb E}[{\bf
z}{\bf z}^T]={\bf 0}$, which means that ${\bf z}$ and ${\bf z}^*$ are
uncorrelated. Note that even in the Gaussian case,
 ${\bf z}$ and ${\bf
z}^*$ can not be independent (they are related by a one-to-one deterministic application) even being uncorrelated.
Thus, the statistical problem of measuring the impropriety of complex-valued vectors is by nature different from standard canonical correlation analysis where the two Gaussian vectors are independent when their canonical correlation coefficients are zero. This emphasizes why state-of-the-art results from multivariate analysis such as canonical correlation analysis, or more recent works to test the independence between complex random vectors
\cite{Klausner2014,Santamaria2017} or the diagonality of the covariance matrix \cite{Mestre2014}, cannot be extended in an easy manner to impropriety testing.}


The original contributions of the present work consist of the following
results in the asymptotic high-dimensional regime: the limiting
distribution of the maximal invariant statistics, and accurate
approximations for standard test statistics used in multivariate analysis,
namely the GLRT and Roy's test, are derived under the null
hypothesis of properness. Moreover, a phase transition
behavior is shown to exist, allowing for the detection of complex-valued low-rank signals (modeled as complex vectors) corrupted by proper noise.

The paper is organized as follows.
In Section~\ref{sec:invstat}, the  maximal invariant statistics are introduced
for {\color{blue} impropriety} testing. Their limiting distributions are also derived.
Section~\ref{sec:testing} provides the exact and limiting distributions of
GLRT, as well as the  limiting distributions of Roy's test statistics.
Some simulations
are conducted in \ref{sec:simu} to appreciate the accuracy of the proposed
approximations. The relation with canonical correlation analysis results
is discussed in Section~\ref{sec:discuss}.
Some concluding remarks are given in the last
section.

\section{\color{blue}  Impropriety coefficients}
\label{sec:invstat}
In order to design tests based on second order statistics and that are not sensitive to reparametrization (by linear transformation), it is known \cite{Andersson:1975,Schreier:2010} that the {\color{blue} impropriety coefficients} should be the quantities to use. In this section, we introduce them and provide joint and marginal distributions in the asymptotic regime for their empirical estimates.

\subsection{Testing problem}
\label{subsec:testprob}
In many applications such as fMRI \cite{Adali:2007}, DOA estimation \cite{Delmas:2004} or communications \cite{Schreier:2010}\footnote{See \cite{Schreier:2010} and references therein for a larger list of applications involving complex-valued signals and {\color{blue} impropriety} related issues.}, it is common use to model the vector of interest, denoted ${\bf z}$, as being {\em improper} and corrupted by {\em proper} Gaussian noise. Consequently, statistical tests have been proposed to investigate the {\color{blue} impropriety} of a complex vector given a sample (of size $M$ in the sequel), from which one needs to decide:
\begin{align}
\begin{cases}
  H_0: & ${\bf z}$ \textrm{ is {\em proper} {\color{blue}iff} condition \eqref{eq:prop} holds,}\\ 
  H_1: & ${\bf z}$ \textrm{ is {\em improper} otherwise.} 
\end{cases}
\label{eq:test}
\end{align}

In order to design  a statistical test which is invariant under linear transformation, and as explained in \cite{Andersson:1975}, one should use the eigenvalues of the augmented covariance matrix ${\bf C}$ given in \eqref{eq:covmat}.
Before introducing the invariant statistics to be derived from observed complex vectors, we first recall some results about the eigenvalues of real augmented PSD (Positive Semi-Definite) matrices.

\subsection{Invariant parameters}
\label{subsec:invparam}
Let $\mathcal{G}$ be the set of non-singular matrices ${\bf G} \in \mathbb{R}^{2N \times 2N}$ s.t. 
\begin{align*}
{\bf G}= \begin{pmatrix}
{\bf G}_1 & -{\bf G}_2\\{\bf G}_2 & {\bf G}_1
\end{pmatrix},
\end{align*}
where ${\bf G}_1,{\bf G}_2 \in\mathbb{R}^{N \times N}$. Let $\mathcal{S}$
be the set of all $2N \times 2N$ real positive definite symmetric matrices.
According to the test formulation \eqref{eq:test} and condition \eqref{eq:prop},
the null hypothesis $H_0$ is equivalent to ${\bf C} \in \mathcal{T}=\mathcal{S} \cap \mathcal{G}$.

As explained in \cite{Andersson:1975}, $\mathcal{G}$ is a group (isomorphic to the group $GL_N(\mathbb{C})$ of non-singular $N \times N$ complex matrices under the mapping ${\bf G} \leftrightarrow {\bf G}_1 + \bmi {\bf G}_2$), with the matrix multiplication as the group operation. Moreover $\mathcal{G}$ acts transitively on $\mathcal{T}$ under the action $({\bf G},{\bf T}) \in \mathcal{G}\times \mathcal{T} \mapsto {\bf G}{\bf T}{\bf G}^{T} \in \mathcal{T}$. Thus, a parametric characterization of $H_0$ should be invariant to this group action:
the value of the parameters to be tested should be the same for ${\bf C}$ and ${\bf G} {\bf C} {\bf G}^{T}$ for any ${\bf G} \in \mathcal{G}$.

Next, we introduce a decomposition for any  ${\bf C} \in  \mathcal{S}$ that was originally given in \cite{Andersson:1975} and reads:
\begin{equation}
{\bf C}=\dot{{\bf C}} + \ddot{{\bf C}}
\end{equation}
where \begin{align*}
\dot{{\bf C}}&=  \frac{1}{2}\begin{pmatrix}
{\bf C}_{{\bf u}{\bf u}} + {\bf C}_{{\bf v}{\bf v}} & {\bf C}_{{\bf u}{\bf v}} - {\bf C}_{{\bf v}{\bf u}}\\
{\bf C}_{{\bf v}{\bf u}} - {\bf C}_{{\bf u}{\bf v}} & {\bf C}_{{\bf u}{\bf u}} + {\bf C}_{{\bf v}{\bf v}}
\end{pmatrix} \in  \mathcal{G},\\
\ddot{{\bf C}}&= \frac{1}{2} \begin{pmatrix}
{\bf C}_{{\bf u}{\bf u}} - {\bf C}_{{\bf v}{\bf v}} & {\bf C}_{{\bf u}{\bf v}} + {\bf C}_{{\bf v}{\bf u}}\\
{\bf C}_{{\bf u}{\bf v}} + {\bf C}_{{\bf v}{\bf u}} & {\bf C}_{{\bf v}{\bf v}} - {\bf C}_{{\bf u}{\bf u}}
\end{pmatrix}.\\
\end{align*}
Using this decomposition, one can define the following $2N\times 2N$ real symmetric matrix
\begin{align}
{\bf \Gamma}({\bf C}) &= \dot{{\bf C}}^{-\frac{1}{2}} \ddot{{\bf C}} \dot{{\bf C}}^{-\frac{1}{2}}.
\label{eq:Gamma}
\end{align}
It is now possible to give the following lemma about the parametrization of ${\bf C}$.
\begin{lemma}[Invariant parameters \cite{Andersson:1975}] Any matrix ${\bf C} \in  \mathcal{S}$ can be written as:
\begin{align*}
{\bf C}= {\bf G} \begin{pmatrix}
{\bf I}_N + {\bf D}_{\lambda} & 0\\0 & {\bf I}_N - {\bf D}_{\lambda}
\end{pmatrix} {\bf G}^T,
\end{align*}
where ${\bf G} \in \mathcal{G}$, ${\bf I}_N$ is the $N\times N$ identity matrix and
${\bf D}_{\lambda}= \mathrm{diag}(\lambda_1,\ldots,\lambda_N)$ is an $N\times N$ diagonal matrix whose diagonal entries denoted as
$\lambda_n$,  for $1 \le n \le N$, are the non-negative eigenvalues of the $2N \times 2N$ matrix ${\bf \Gamma}({\bf C})$ given in \eqref{eq:Gamma}.
They  satisfy the following  properties:
\begin{enumerate}
 \item $\lambda_n$ and $-\lambda_n$, for $1 \le n \le N$, form the set of eigenvalues of
${\bf \Gamma}({\bf C})$,
 \item $\lambda_n \in [0,1]$  with, by convention, the following ordering $1 \ge \lambda_1 \ge \ldots \ge \lambda_N \ge 0$.
\end{enumerate}
\label{lem:rep}
\end{lemma}

\begin{proof}
See \cite[lemma 5.1 and 5.2]{Andersson:1975}.
\end{proof}

Lemma \ref{lem:rep} shows that any invariant parameterization of the covariance matrix ${\bf C}$ for the group action
of $\mathcal{G}$ depends only on the $N$ (non-negative) eigenvalues $1 \ge \lambda_1 \ge \ldots \ge \lambda_N \ge 0$
of ${\bf \Gamma}({\bf C})$. Thus these eigenvalues are termed {\em maximal invariant parameters} \cite[Chapter 6]{Lehmann:2005}. Moreover under the null hypothesis $H_0$, one has that $\lambda_1 = \ldots = \lambda_N=0$ as $\ddot{{\bf C}}$ reduces to the zero matrix according to \eqref{eq:prop}. Within the invariant parameterization, the testing problem in \eqref{eq:test} becomes:
\begin{align}
\begin{cases}
  H_0:& \lambda_1 = 0,\\
  H_1: & \lambda_1 > 0,
\end{cases}
\label{eq:testinv}
\end{align}
where the alternative hypothesis $H_1$
means that there exists at least one positive eigenvalue.
Note that the invariance property ensures that the test does not depend on the (common) representation basis of the real and imaginary parts of ${\bf z}$, {\em i.e.} vectors $\bf u$ and $\bf v$. Also, $\lambda_n$, the eigenvalues of $\Gamma({\bf C})$, are directly related to the ones obtained using the {\em complex augmented representation}, {\em i.e.} based on the complex covariance matrix:
\begin{align*}
 {\mathbb E}[\tilde{\bf z}\tilde{\bf z}^{\dagger}]
 &=
\begin{pmatrix}
{\bf C}_{{\bf z}{\bf z}} & {\bf C}_{{\bf z}{\bf z}^*} \\
{\bf C}_{{\bf z}^*{\bf z}} & {\bf C}_{{\bf z}^*{\bf z}^*}\\
\end{pmatrix},
\end{align*}
where $.^{\dagger}$ stands for transposition and conjugation,
and ${\bf C}_{{\bf r}{\bf s}}= {\mathbb E}[ {\bf r} {\bf s}^{\dagger}]$ denotes the (complex) cross-covariance between the sized $N$ complex vectors ${\bf r}$ and ${\bf s}$. In fact, as detailed in  \cite[Chapter 3]{Schreier:2010}  and \cite{Walden:2009, hellings2019},
the eigenvalues $\lambda_n$, for $1\le n\le N$, are also the square roots of the eigenvalues of the following $N \times N$ complex matrix:
\begin{align}
 {\bf C}_{{\bf z}^*{\bf z}^*}^{-1}  {\bf C}_{{\bf z}^*{\bf z}} {\bf C}_{{\bf z} {\bf z}}^{-1} {\bf C}_{{\bf z} {\bf z}^*}
\label{eq:CCA}
 \end{align}
Matrix \eqref{eq:CCA} corresponds to the usual population canonical correlation matrix to derive the canonical variables between two vectors,
here the complex ones ${\bf z}$ and ${\bf z}^*$.
{\color{blue} Hence the squared eigenvalues $\lambda_n^2$ are commonly referred to as the canonical correlation coefficients. In our specific framework of impropriety testing, these coefficients are also known as {\em circularity coefficients} \cite{Eriksson:2006,Moreau:2013}, or
{\em impropriety coefficients} \cite{hellings2019}.
In the sequel, to better emphasize the statistical differences between our impropriety testing problem with the problem of testing the independence between two vectors\footnote{\color{blue}Note again that even under $H_0$ where the ${\bf z}$ and ${\bf z}^*$ are uncorrelated, {\em i.e.} $\lambda_1= \cdots =\lambda_N=0$,  they cannot be independent since they are deduced from one another in a deterministic way.
The probability measure of the impropriety test statistics will therefore be different from that of the independent case.}, we will make use of the name {\em  population impropriety coefficients} for the eigenvalues $\lambda_n$. The sample version of these eigenvalues (derived from the sample covariances) will be will be denoted as $l_n$ and $r_n \equiv l_n^2$ for their squared values,  and referred to as the {\em sample  impropriety coefficients} and the  {\em sample squared  impropriety coefficients} respectively.}


\subsection{Invariant statistics}
\label{subsec:statistics}
Consider a sample of size $M$, denoted ${\bf X}=\left\{ {\bf x}_m \right\}_{m=1}^M$, where $ {\bf x}_m= [{\bf u}^T_m , {\bf v}^T_m]^T$ are  $2N$-dimensional i.i.d. Gaussian real vectors  with zero mean and covariance matrix ${\bf C}$.
In the Gaussian framework, a sufficient statistics is given by the $2N \times 2N$ sample covariance matrix:
\begin{align}
{\bf S}&= \begin{pmatrix}
{\bf S}_{{\bf u}{\bf u}} & {\bf S}_{{\bf u}{\bf v}} \\
{\bf S}_{{\bf v}{\bf u}} & {\bf S}_{{\bf v}{\bf v}}\\
\end{pmatrix},
\end{align}
with ${\bf S}_{{\bf a}{\bf b}} \in {\mathbb R}^{N \times N}$ the real-valued sample (cross)covariance matrix of real vectors $\left\{ {\bf a}_m \right\}_{m=1}^M$ and $\left\{ {\bf b}_m \right\}_{m=1}^M$ such that:
\begin{equation}
{\bf S}_{{\bf a}{\bf b}}= \frac{1}{M} \sum_{m=1}^M {\bf a}_m{\bf b}_m^T.
\end{equation}
We assume here that $M \ge 2N$, thus ${\bf S}$ belongs to the
real symmetric positive definite matrices set $\mathcal{S}$.
According to the previous section, since $H_0$ is invariant under the action of the group $\mathcal{G}$, an invariant test statistic must only depend on the $N$ non-negative eigenvalues $l_n$, $1\le n \le N$,  of
\begin{align}
 {\bf \Gamma}({\bf S})= \dot{{\bf S}}^{-\frac{1}{2}} \ddot{{\bf S}} \dot{{\bf S}}^{-\frac{1}{2}}.
 \label{eq:GammaS}
\end{align}
These {\em sample impropriety coefficients} obey $1 \ge l_1 \ge \ldots \ge l_N \ge 0$ according to Lemma \ref{lem:rep}, and are an estimate of the population impropriety coefficients $\lambda_n$
obtained from the population covariance ${\bf C}$. Note that all the $\lambda_n$ are zero under the null hypothesis $H_0$, and at least one is non-negative otherwise. As a consequence, the distribution of the $l_n$ should be stochastically greater under $H_1$ than under $H_0$. All invariant test can be derived from this property. A key point to derive now a tractable statistical test procedure is to characterize the null distribution of these sample impropriety coefficients.


\subsection{Eigenvalue distribution under $H_0$}

Let $\mathcal{B}_N(\frac{1}{2} n_1,\frac{1}{2}  n_2)$ denote the $N \times N$-dimensional matrix variate beta distribution
with parameters $n_1$ and $n_2$ as defined for instance in \cite[definition 3.3.2, p. 110]{Muirhead:2005}. It is possible to obtain, under $H_0$, the joint probability density function (pdf) of the squared eigenvalues of $\Gamma({\bf S})$ in terms of this matrix-variate beta distribution.
\begin{proposition}[Joint distribution of impropriety coefficients]
Under $H_0$, the vector $(r_1,\ldots,r_N)$ of the sample squared impropriety coefficients $r_n \equiv l_n^2$, for $1\le n \le N$, is distributed as the eigenvalues
of the matrix-variate beta distribution $\mathcal{B}_N(\frac{1}{2} n_1,\frac{1}{2}  n_2)$, with parameters $n_1=N+1$ and $n_2= M-N$. Moreover,
the joint pdf of $(r_1,\ldots,r_N)$ is expressed as:
\begin{align}
p(r_1, \ldots,r_N) \propto \prod_{n=1}^N  (1- r_n)^{(M-2N-1)/2} \prod_{k<n}^N (r_k - r_n),
\label{eq:pdfBeta}
\end{align}
where $1\ge r_1 \ge \ldots \ge r_N \ge 0$.
\label{prop:pdfBeta}
\end{proposition}

\begin{proof}
As shown in \cite[pp. 39-41]{Andersson:1975}, the sample eigenvalue vector $(l_1,\ldots,l_N)$ is characterized by the following pdf:
\begin{align*}
p(l_1, \ldots,l_N) \propto \prod_{n=1}^N (2 l_n) (1- l_n^2)^{(M-2N-1)/2} \prod_{k<n}^N (l_k^2 - l_n^2).
\end{align*}
A simple change of variables yields the pdf of  $(r_1,\ldots,r_N)$ given in \eqref{eq:pdfBeta}.
Moreover, according to \cite[Theorem 3.3.4, p. 112]{Muirhead:2005}, \eqref{eq:pdfBeta} is the pdf of the eigenvalues of the matrix variate beta distribution $\mathcal{B}_N(\frac{N+1}{2} ,\frac{M-N}{2} )$, which concludes the proof.
\end{proof}

It is interesting to note that the pdf given in Proposition \ref{prop:pdfBeta} is close to what would be obtained if one would perform a canonical correlation analysis on ${\bf z}$ and ${\bf z}^*$ considered as $N$-dimensional real Gaussian independent vectors. Here vectors ${\bf z}$ and ${\bf z}^*$ are actually complex valued and fully dependent. This is further discussed in Section \ref{sec:discuss}.

Expression \eqref{eq:pdfBeta} gives, under the $H_0$ hypothesis, the joint distribution of the squared sample impropriety coefficients $(r_1,\ldots,r_N)$. In the general case, obtaining an analytic expression of marginal distributions of individual eigenvalues is a complicated task. However, in the asymptotic regime, {\em i.e.} when the dimension $N$ and the number of samples $M$ go to infinity while their ratio stays commensurable,
one can obtain those marginal laws. The following theorem gives
the distribution of one {\color{blue} (unordered)} sample impropriety coefficient in this regime.

\begin{theorem}[Limiting empirical distribution]
As $M,\,N\, \rightarrow  \, \infty$ with the ratio
$M/N \, \rightarrow \, \gamma \in [2,+\infty)$ being finite, the marginal empirical distribution of the {\color{blue} unordered} sample squared impropriety coefficients ({\em i.e.} the squared eigenvalues of ${\bf \Gamma}({\bf S})$) converges, under the $H_0$ hypothesis, to the probability measure with density:
\begin{align}
f(r)&=  \frac{1} { 2 \pi (1-r)  }   \sqrt{ 4 (\gamma-1) \tfrac{1-r}{r} - (\gamma-2)^2  },
\label{eq:pdfasympt}
\end{align}
on its support  $r \in (0, c)$, with $c= \tfrac{4(\gamma-1)}{\gamma^2} \in (0,1]$.
\label{thm:pdfev}
\end{theorem}
\begin{proof}
See Appendix \ref{app:pdf}.
\end{proof}

\begin{Corollary}[Moments]
The mean and variance of the limiting distribution under $H_0$  of a sample squared impropriety coefficients are expressed
respectively as ${1}/{\gamma}$ and
${(\gamma-1)}/{\gamma^3}$.
\end{Corollary}
\begin{proof}
Expressions of these limiting moments can be derived directly from the pdf \eqref{eq:pdfasympt}
by symbolic computation.
\end{proof}

A few remarks are in order:
\begin{itemize}
 \item When $\gamma \rightarrow +\infty$, the expression of the mean and variance emphasizes that the sample impropriety coefficients converge to zero, which are the population values under $H_0$. This is the usual behavior in small dimension when $N$
 is fixed while $M$ tends to infinity.
 \item Conversely, in the special case where $\gamma=2$, the asymptotic null distribution of the squared sample impropriety coefficients is $\mathcal{B}(\frac{1}{2} ,\frac{1}{2})$, known as the arcsine law. In this limiting case,
the sample impropriety coefficients are symmetrically distributed on $[0,1]$
around $1/2$  with two symmetric modes at the edges (even if the population
impropriety coefficients are zero).
\end{itemize}


\section{Testing for {\color{blue} impropriety}}
\label{sec:testing}
In this section, we make use of the results from Proposition \ref{prop:pdfBeta} and Theorem \ref{thm:pdfev} to introduce the asymptotic behavior of two {\color{blue} impropriety} tests: the classical GLRT and Roy's test (based on the largest eigenvalue of the ${\bf \Gamma}({\bf S})$ matrix).
\subsection{GLRT}
\label{subsec:glrt}
\subsubsection{Expression of the GLRT statistic}
A very classical procedure to test for {\color{blue} impropriety} is obtained from the Generalized Likelihood Ratio Test (GLRT) statistic defined as:
\begin{align*}
T \propto \frac{\stackrel[{\bf C} \textrm{ s.t. $H_0$}]{\sup} \quad p({\bf X}\, ; \, {\bf C}  )}{\stackrel[ {\bf C} \textrm{ s.t.   $H_1$}]{\sup}\quad p( {\bf X}\, ; \,{\bf C} )},
\end{align*}
where $p({\bf X}\, ; \,{\bf C} )$ is the multivariate normal pdf of the sample ${\bf X}$ composed of $M$  i.i.d. $2N$-dimension real Gaussian vectors  with zero mean and covariance matrix ${\bf C}$. Under $H_1$, ${\bf C} \in \mathcal{S}$ is a symmetric definite positive matrix. Its maximum likelihood (ML) estimate is the sample covariance ${\bf S}$. Under $H_0$, one has that ${\bf C}= \dot{{\bf C}}$ so that ${\bf C} \in \mathcal{T}$. Then the ML estimate of ${\bf C}$ under $H_0$ reduces to $\dot{{\bf S}}$, as shown for instance in \cite{Andersson:1975}. The testing problem in \eqref{eq:test} can thus be rephrased:
\begin{align}
\begin{cases}
  H_0:& {\bf C} \in \mathcal{T},\\
  H_1: & {\bf C} \in \mathcal{S}.
\end{cases}
\label{eq:testgroup}
\end{align}
Actually, the GLRT statistic is expressed as:
\begin{align}
T &= |{\bf S} | / | \dot{{\bf S}} |, \nonumber \\
&=  |\dot{{\bf S}}^{\frac{1}{2}} \left( {\bf I}_{2N} +   {\bf \Gamma}({\bf S}) \right) \dot{{\bf S}}^{\frac{1}{2}} | / |\dot{{\bf S}}|
=\left| {\bf I}_{2N} + \Gamma({\bf S}) \right|, \nonumber  \\
&= \prod_{n=1}^N (1 + l_n)(1-l_n) = \prod_{n=1}^N (1 - r_n),
\label{eq:GLRT}
\end{align}
where the first line is due to the Gaussian pdf expression; the second line
comes from the decomposition
${\bf S}=\dot{{\bf S}} + \ddot{{\bf S}}$ and the  expression \eqref{eq:GammaS} of ${\bf \Gamma}({\bf S})$; the third line
comes from Lemma \ref{lem:rep}, where $r_n= l_n^2$, $1 \le n \le N$, are the sample squared  impropriety coefficients. As explained previously, it is important to note that the GLRT is invariant:
the resulting statistics given in \eqref{eq:GLRT} only depends on the eigenvalues of ${\bf \Gamma}({\bf S})$.

\subsubsection{Distribution under the  hypothesis $H_0$}

Let $\Lambda(d,m,n)$ denote Wilks lambda distribution, with dimension parameter $d$ and degrees of freedom parameters $m$ and $n$, as defined for instance in \cite[definition 3.7.1, p. 81]{Mardia:1979}.
\begin{theorem} The GLRT statistics $T$ given in \eqref{eq:GLRT} is
distributed under $H_0$ as the following Wilks lambda distribution:
\begin{align*}
T \sim \Lambda(N,M-N,N+1).
\end{align*}
Moreover this statistics can be expressed under $H_0$ as:
\begin{align}
T= \prod_{n=1}^N u_n,
\label{eq:Tbeta}
\end{align}
where the $u_n$ are independent beta-distributed random variables such that
$u_n \sim \mathcal{B}\left( \frac{M-N-n+1}{2}, \frac{N+1}{2} \right)$, for $1\le n \le N$.
\label{thm:GLRT}
\end{theorem}
\begin{proof}
According to Proposition \ref{prop:pdfBeta},  the $r_n$  in \eqref{eq:GLRT} are distributed as the eigenvalues of the matrix variate beta distribution $\mathcal{B}_N(\frac{1}{2} n_1,\frac{1}{2}  n_2)$ with parameters $n_1=N+1$ and $n_2= M-N$. Using the mirror symmetry property
of the beta distribution, the $(1 - r_n)$ are distributed as the eigenvalues of the random matrix ${\bf U} \sim \mathcal{B}_N(\frac{1}{2} n_2,\frac{1}{2}  n_1)$.
According now to \cite[Theorem 3.3.3, p. 110]{Muirhead:2005}, ${\bf U}$ can be decomposed as ${\bf U}={\bf \Theta}^T {\bf \Theta}$ where ${\bf \Theta}$ is upper triangular with diagonal entries $\theta_{nn}$ that are independent and where $u_n \equiv \theta_{nn}^2 \sim \mathcal{B}\left(\frac{n_2-n+1}{2},\frac{n_1}{2} \right)$ for $1\le n \le N$. This concludes the proof since $T=|{\bf U}|= \prod_{n=1}^N\theta_{nn}^2$.
\end{proof}

Equation \eqref{eq:Tbeta} gives also a more efficient way to sample from the null distribution of $T$ in $O(N)$ independent draws, as it is actually not required to generate the $2N \times 2N$ sample covariance matrix ${\bf S}$, nor to compute the eigenvalues of ${\bf \Gamma}({\bf S})$.


\subsubsection{High-dimensional asymptotic distribution under $H_0$}

The characterization  given in \eqref{eq:Tbeta} allows us to derive, under the null hypothesis $H_0$,
an asymptotic distribution for the GLRT statistic $T$ in the high dimensional ({\em i.e.} large $N$) case.
This yields a simple tractable closed form approximation of the considered Wilks lambda distribution when both the dimension $N$ and the sample size $M$ are large.
\begin{theorem}[Central limit theorem in high dimension]
Let  $T'= - \ln T$ where $T$ is the GLRT statistic given in \eqref{eq:GLRT}. Assume that $M,\,N\, \rightarrow  \, \infty$  so that the ratio $M/N \, \rightarrow \, \gamma \in (2,+\infty)$. Under $H_0$, the following asymptotic normal distribution is obtained for $T'$:
\begin{align}
 \frac {1}{s} \left( T' - m \right) \ {\xrightarrow {d}} \ \mathcal{N}(0,1)
\label{eq:normasympt}
 \end{align}
where
\begin{align*}
m &= M \left[ \ln \frac{\gamma}{\gamma-1} + \frac{\gamma-2}{\gamma}  \ln \frac{\gamma-2}{\gamma-1} \right]
+ \frac{1}{2} \ln \frac{\gamma}{\gamma-2} ,\\
s^2 & = 2 \left[ \ln \frac{(\gamma-1)^2}{\gamma(\gamma-2)} + \frac{1}{M} \frac{1}{\gamma-2} \right].
\end{align*}
\label{thm:asympt}
\end{theorem}
\begin{proof}
See Appendix \ref{app:asympt}.
\end{proof}
Bartlett derived a classical approximation for Wilks lambda distribution \cite[p. 94]{Mardia:1979} in a low-dimensional setting. This gives, when the dimension $N$ is fixed while $M$ goes to infinity, the same asymptotic distribution as obtained in \cite{Walden:2009}:
\vspace{-1mm}
\begin{align}
 -(M-N) \ln T   \ {\xrightarrow {d}} \ \chi^2_{N(N+1)},
 \label{eq:Bart}
\end{align}
where $ \chi^2_{N(N+1)}$ denotes the chi-squared distribution with $N(N+1)$ degrees of freedom. Note that this approximation was used recently in \cite{Hasija2017}, in a high-dimensional setting in the absence of better approximation.

Using Theorem \ref{thm:asympt}, the Bartlett  approximation can now be
adjusted
to cover both the low and high-dimensional cases.
Let $\mathcal{G} \left( q, p \right)$ denote the gamma distribution with pdf
$f(x) \propto x^{q-1}e^{-x/p}$, where $q$ and $p$ are the shape and scale
parameters respectively.
\begin{Corollary}[Adjusted Bartlett approximation]
Let $\gamma=M/N >2$,
then the log-GLRT statistics $T'=- \ln (T) $ can be approximated as a shifted gamma distribution:
\begin{align}
 \frac{1}{s} \left(T' - \alpha \right) \approx \mathcal{G} \left( q, p \right),
 \label{eq:shiftgam}
\end{align}
with
\begin{align*}
 q &= N(N+1)/2, \quad
 p = \sqrt{ 1 / q }, \quad
 \alpha = m - pq s,
\end{align*}
and $m$ and $s^2$ are defined in Theorem \ref{thm:asympt}, and where $\approx$
stands for pointwise equivalence of distribution functions for large $M$
under both the low dimensional, {\em i.e.} $N$ is fixed and small w.r.t. $M$,
or the high dimensional, {\em i.e.} $N$ has order of $M$, regime.
\label{cor:glrtgam}
\end{Corollary}
\begin{proof} In the high dimensional setting, under the assumptions of Theorem
\ref{thm:asympt},
the gamma distribution $\mathcal{G} \left( q, p \right)$ converges towards the normal one as the shape parameter $q$
goes to infinity. Since the mean and variance of $T'$ are the same for the
shifted gamma approximation \eqref{eq:shiftgam}  and the normal one
\eqref{eq:normasympt}, they are asymptotically equivalent.

In the low-dimensional setting where $N$ is fixed while $M$ goes to infinity,
one gets that $p$ and $q$ are fixed, $\alpha/s \rightarrow 0$,  and
$p s =
\frac{2}{N(\gamma-1)} +
O\left( 1/\gamma^2 \right) = \frac{2}{M-N} + O(1/M^2)$. Moreover, in this limiting case, $T' {\xrightarrow {d}} 0$
according to the decomposition given in \eqref{eq:Tbeta}.
Thus $\frac{1}{s} \left(T' - \alpha \right) = \frac{p}{2} (M-N)T' + o(1)$.
 Because $\chi^2_{q} = \frac{2}{p}  \mathcal{G} \left( q, p \right)$, \eqref{eq:shiftgam} means that
$(M-N) T'$ is asymptotically $\chi^2_{q}$ distributed. This is the Bartlett
limiting distribution \eqref{eq:Bart}, which is known to be valid in this
low dimensional asymptotic regime
\end{proof}

\subsection{Roy's test}
\label{subsec:roy}
In multivariate statistics, Roy's test is a well known procedure to detect the alternate hypothesis $H_1$ for which at least one eigenvalue is non-zero. This test relies on the statistics of the largest eigenvalue \cite[p. 84]{Mardia:1979} or, equivalently in our case, the statistics of the largest squared impropriety coefficients $r_1= l_1^2$. The principle is to reject the $H_0$ hypothesis as soon as $r_1 > \eta_{\alpha}$,  where the threshold $\eta_{\alpha}$ is tuned according to the law of $r_1$ under the $H_0$ hypothesis together with the nominal control level $\alpha$ (probability of false alarm).

\begin{theorem}[Limiting null distribution for Roy's test]
As $M,\,N\, \rightarrow \, \infty$ such that the ratio $M/N \, \rightarrow \, \gamma \in[2,+\infty)$ is finite, let $W= \log{( r_1 / (1-r_1) )}$ be the {\em logit} transform of the largest impropriety coefficient $r_1$. Under $H_0$, the asymptotic law of $W$
converges towards a first order Tracy-Widom law denoted as $\mathcal{TW}_1$:
\begin{align}
 \frac{W - \mu}{\sigma} \, \rightarrow \mathcal{TW}_1,
 \label{eq:W}
\end{align}
with
\begin{align*}
    \mu &= 2\log \tan \left( \tfrac{\phi +\psi}{ 2} \right),\\
    \sigma^3 &= \frac{16}{ M^2 }  \frac{1}{ \sin^2 (\phi+\psi)  \sin \psi \sin \phi },\\
    \psi &= \arccos{ \left( \tfrac{M-2N+1}{M} \right) },\\
    \phi &= \arccos{ \left( \tfrac{M-2N-1}{M} \right) }.
\end{align*}
\label{thm:Roy}
\end{theorem}
\begin{proof}
This is a direct result of proposition \ref{prop:pdfBeta} and the asymptotic law of the largest eigenvalue of a  $\mathcal{B}_N(\frac{1}{2} n_1,\frac{1}{2} n_2)$ matrix-variate distribution
given in \cite{johnstone2009}.
\end{proof}
The variable $W$ being expressed as an increasing function of  $r_1$, Roy's test is equivalent to $W$ and the Theorem \ref{thm:Roy} allows for the {\color{blue} calibration} of the test.
It should be noted that \cite{chiani2016} proposes a procedure to evaluate the exact (not asymptotic) law of $W$.
Nevertheless, the simple approximation by the $\mathcal{TW}_1$ law is in practice sufficiently precise for most cases as long as the dimension is large enough (e.g. typically for $N \ge 10$).

\subsection{Spiked impropriety model}

{\color{blue} Recently, spiked models involving complex valued vectors have been found to nicely describe multiplexed phase retrieval problems in complex media imaging \cite{dong:2019}. Here, we provide theoretical results related to the phase transition behavior occurring in such models.}
Spiked models are special sparse cases for the alternative hypothesis $H_1$. They assume that the rank of the
population matrix is low and remains fixed in the high-dimensional asymptotic regime.
For the {\color{blue} impropriety} test setting, this means that the number $k$ of non-zero
eigenvalues of $\bf \Gamma(\bf S)$, or equivalently $\ddot{\bf C}$, is fixed. An
example, which corresponds to a low-rank improper signal corrupted by proper
noise, is given in Section \ref{sec:simuspike}.

\begin{theorem}[Phase transition threshold]
Assume that there exist $k$ non-zero population impropriety coefficients $\lambda_1 \ge  \cdots \ge \lambda_k >0$,
and $\lambda_{k+1} = \cdots = \lambda_N= 0$
 where $k$
is fixed.
Under the assumptions of Theorem \ref{thm:pdfev}, we have the following convergence for
the square of the largest impropriety coefficient $r_n$ with $1 \le n \le k$,
\begin{align*}
  \textrm{ if } \lambda_n^2 \le \rho_c, \qquad r_n & \stackrel{a.s.}{\longrightarrow} c,\\
  \textrm{If } \lambda_n^2 > \rho_c, \qquad r_n & \stackrel{a.s.}{\longrightarrow} \overline{\rho}_n,
\end{align*}
where
\begin{align}
\rho_c & = \frac{1}{\gamma-1}, \quad \overline{\rho}_n =\lambda_n^2 \left( \frac{\gamma-1}{\gamma} + \frac{1}{\gamma \lambda_n^2}  \right)^2,
\label{eq:rhoc}
\end{align}
are respectively the phase transition threshold and the limiting values,
and $c$ is the edge of the limiting distribution of the bulk defined in Theorem
\ref{thm:pdfev}.
\label{thm:spike}
\end{theorem}
\begin{proof}
This follows from Proposition \ref{prop:pdfBeta} and from results for high dimensional limiting distribution of spiked models with
Beta distributed matrices given in \cite[see Theorem 1.8]{bao2019}, or   \cite{Johnstone2018}.
\end{proof}

Theorem \ref{thm:spike} shows that when the spikes are weaker than a given phase transition threshold $\rho_c$,
none of the sample impropriety coefficients separate from the bulk.
This makes the testing problem challenging, and Roy's test would be  powerless in this case. Conversely, for spikes $\lambda_n^2$ larger than $\rho_c$, is is easy to check that
$\overline{\rho}_n > c$. These sample impropriety coefficients separate now from the bulk, and Roy's test is expected to be very powerful.

\section{Simulations}
\label{sec:simu}
This section starts with simulations validating the accuracy of the asymptotic distributions derived under the properness hypothesis $H_0$.
Then, {\color{blue} impropriety} testing is illustrated under different alternative $H_1$ hypotheses: {\em i)} {\em equi-correlated} model
($N$ non-zero identical impropriety coefficients $\lambda_1=\cdots=\lambda_N = \rho> 0$), {\em ii)} spiked model
($\lambda_1 > 0$ while $\lambda_2=\cdots=\lambda_N = 0$) {\color{blue}and {\em iii)} a mixed model (first half of impropriety coefficients $\lambda_1,\ldots,\lambda_{[N/2]}>0$ gradually decreases,  and the other half is zero).} {\color{blue} Note that, in the following simulations, we make use, for the sake of readability, of the following abuse of notation: $\gamma$ will now denote the ratio $M/N$ in the considered approximations rather that its limit.}

\subsection{Empirical distribution of impropriety coefficients}
\label{subsec:adequation}

\subsubsection{Empirical vs limiting distributions}
\label{subsubsec:Eigen_EmpLaw}
{\color{blue} In order to illustrate the accuracy of Theorem \ref{thm:pdfev} in various settings (including small vector dimensions),} Fig. \ref{fig:hist} displays, for different values of $N$ and  $\gamma$, the empirical distribution of the squares of the sample impropriety coefficients
under the properness assumption $H_0$.
This shows the very good agreement with the limiting empirical distribution derived in Theorem \ref{thm:pdfev}.
Note that, when $N=10$, small fluctuations can be observed ({\color{blue}gray} bars) around the right edge $c$ of the limiting empirical distribution.
But in a larger dimension ($N=100$), the greatest coefficients converged well towards this edge.

\begin{figure}[htbp!]
    \begin{center}
    \includegraphics[width=0.99\linewidth]{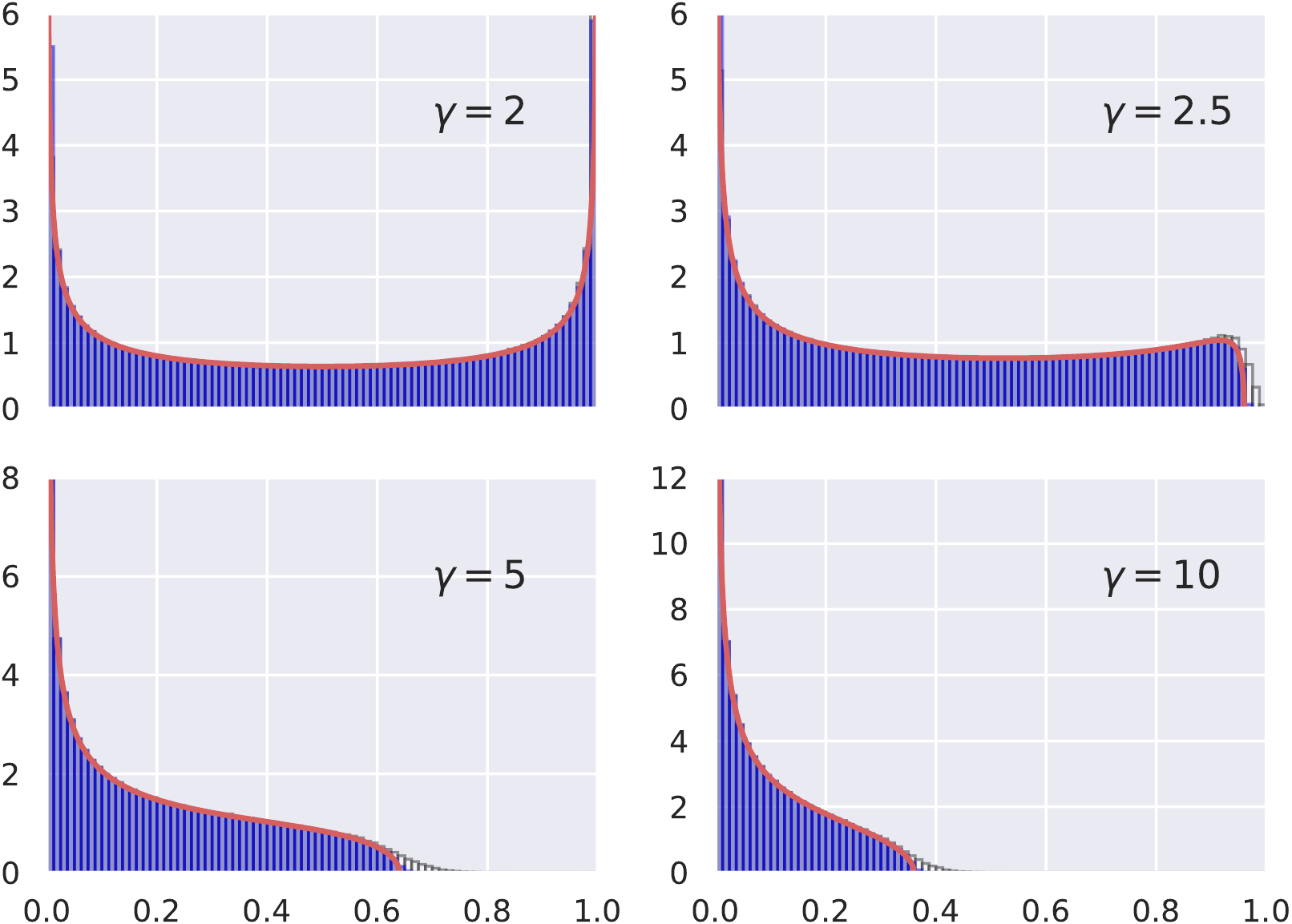}
    \end{center}
    \caption{Histograms of the squared sample impropriety coefficients under
    $H_0$ for different values of $\gamma$: blue bars are for $N=100$, {\color{blue}gray}
bars are for $N=10$. The pdf of the limiting distribution given in theorem
 \ref{thm:pdfev} is shown in  solid red line. Histograms obtained with $1000$
Monte-Carlo runs.
    \label{fig:hist}
    }
\end{figure}

\subsubsection{Distribution of the GLRT statistic}
\label{subsubsec:lawGLRT}

Fig. \ref{fig:pp} depicts, for different values of $M$ and  $\gamma$, a probability-probability plot of the theoretical null distribution of $T$ against each one of these asymptotic approximations.
A deviation from the $y=x$ line indicates a difference between the theoretical and the asymptotic distributions.
This shows that, as expected for high-dimensional setting (e.g.,  $\gamma \le 5$) and/or large sample sizes (e.g., $M \ge 1000$),
the asymptotic log-normal  distribution derived in {\color{blue} Theorem \ref{thm:asympt}} becomes very accurate and much better than the Bartlett approximation.
In addition, the adjusted Bartlett approximation obtained in Corollary \ref{cor:glrtgam} is very accurate in all cases
(low/high-dimension or small/large sample size).
{\color{blue}
This latter corrects and generalizes the classical Bartlett approximation that may behave poorly even for small (e.g. $N=4$) dimensional vectors, as we can see in the top-left subplot ($M=10$, $\gamma=2.5$) of Fig. \ref{fig:pp}.}
{The adjusted approximation is therefore of practical interest to calibrate the GLRT procedure according to a nominal significance level. 
\begin{figure}[ht!]
\centering
    \includegraphics[width=\columnwidth]{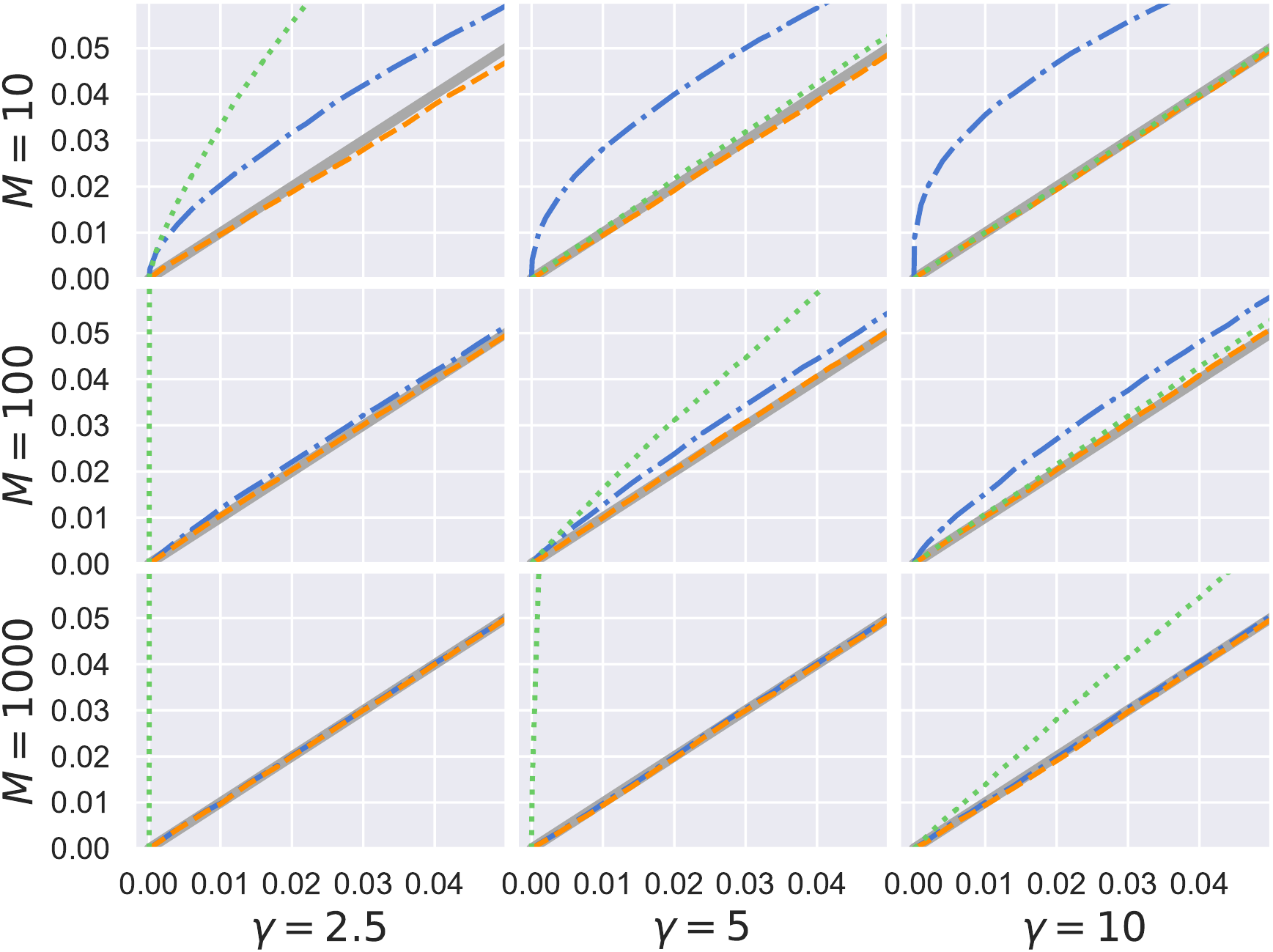}
    \caption{Comparison, for different values of $\gamma$ and $M$, of
asymptotic approximations for the GLRT statistics $T$:
    $\Pr(T> q_\alpha)$ under $H_0$ vs the nominal control level $\alpha$ in $[0,0.05]$ where
    $q_\alpha$ is the $1-\alpha$
    quantile either for the log-normal approximation given in Theorem \ref{thm:asympt},
    shown in dashdotted blue line, or the adjusted Bartlett approximation given in Corollary \ref{cor:glrtgam},
    shown in dashed orange line,
    or the standard Bartlett approximation \eqref{eq:Bart}, shown in dotted green line.
The solid gray line represents the $y=x$ values. Probabilities estimated with
$10^6$ Monte-Carlo runs.}
\label{fig:pp}
\end{figure}

\subsubsection{Distribution of the Roy's statistic}
\label{subsubsec:lawmaxeig}

Fig. \ref{fig:pp_roy} depicts, for different values of $N$ and  $\gamma$, a probability-probability plot of the
theoretical null distribution for the largest impropriety coefficients statistics $r_1$ , or equivalently its logit-transform  $W$,
against the asymptotic approximation given in Theorem \ref{thm:Roy} .
%
This shows that even for a moderate dimension ($N=10$), the Tracy-Widom approximation is quite accurate, and becomes
very accurate for a larger dimension ($N=100$).


\begin{figure}[htbp!]
\centering
\includegraphics[width=\columnwidth]{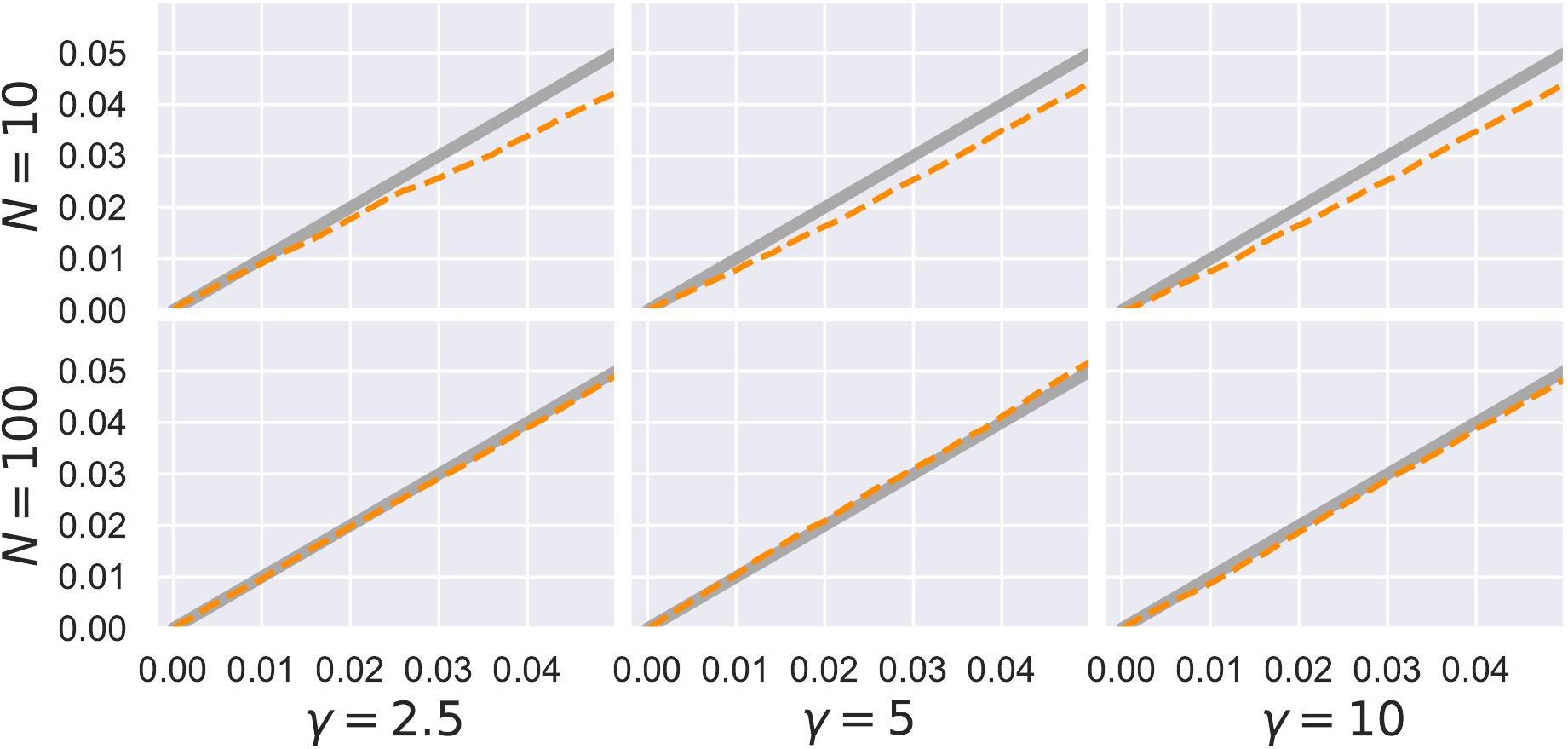}
    \caption{Comparison, for different values of $\gamma$,  of asymptotic
approximations for the Roy's statistics $W$:
    $\Pr(W> q_\alpha)$ under $H_0$ vs the nominal control level $\alpha$ in $[0,0.05]$ where
    $q_\alpha$ is the $1-\alpha$
    quantile for the Tracy-Widom approximation given in Theorem \ref{thm:Roy},
    shown in dashed orange line.
    The solid gray line represents the $y=x$ values.  Probabilities estimated with $10^6$ Monte-Carlo runs.}
\label{fig:pp_roy}
\end{figure}

\subsection{Some {\color{blue} impropriety} tests scenarios}
\label{subsec:2tests}

\subsubsection{Equal impropriety  coefficients}
\label{subsec:equicorr}
The case where the population impropriety coefficients are all non-zero and equal, hereinafter referred to as {\em equi-correlated} model, can be obtained
when the real and imaginary parts have a common contribution:
\begin{align}
\begin{split}
 {\bf u}_{m} &= {\bf s }_{m} + \sqrt{\theta} {\bf q}_m,\\
 {\bf v}_{m} &= {\bf t }_{m} + \sqrt{\theta} {\bf q}_m,
\end{split}
\label{eq:modequi}
\end{align}
where $\theta>0$, ${\bf s }_{m}$, ${\bf t }_{m}$  and ${\bf q}_{m}$ are i.i.d. Gaussian vectors in $\mathbb{R}^N$, for $1 \le m \le M$.
Straightforward computations show that the non-negative  roots of ${\bf \Gamma( {\bf C} )}$, i.e the
 population impropriety coefficients $\lambda_1, \ldots, \lambda_N$, are all equal to $\lambda \equiv \frac{\theta} {1 + \theta}$.

\begin{figure}[htbp!]
    \begin{center}
    \includegraphics[width=0.8\columnwidth]{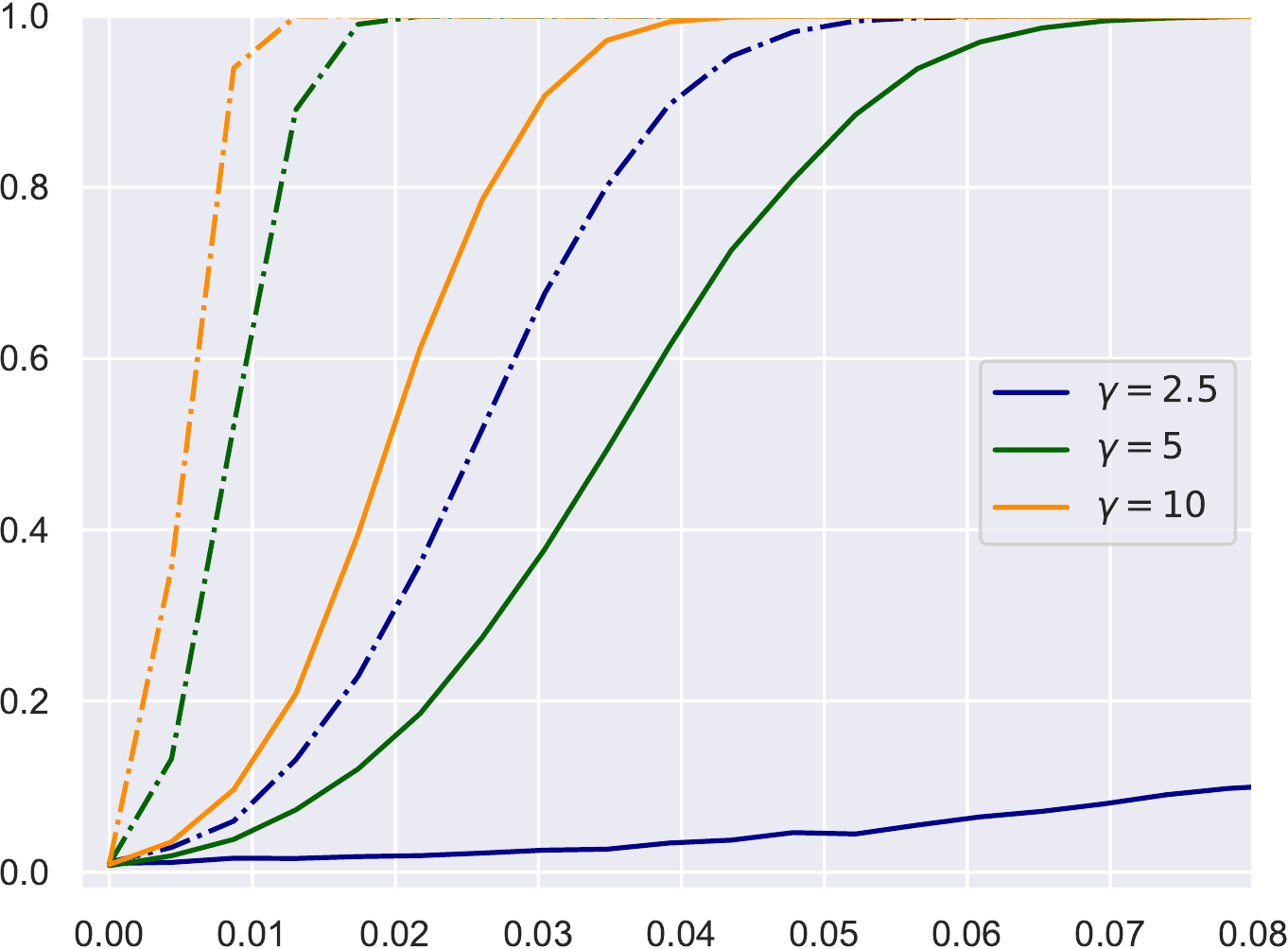}
    \end{center}
    \caption{Power of Roy's test, solid line, and GLRT, dashdotted line, vs
    the squared population impropriety coefficients $\lambda^2$  under the
    equi-correlated model described in \eqref{eq:modequi}, for different values of $\gamma$:
    blue lines for $\gamma=2.5$, green for $\gamma=5$, orange for $\gamma=10$
    (dimension $N=100$, nominal value of false alarm probability $\alpha=0.01$).
    Powers estimated with $1000$ Monte-Carlo runs.
    }
    \label{fig:power}
\end{figure}

Fig. \ref{fig:power} displays the power of both GLRT and Roy's test, under the alternative $H_1$ obtained for this equi-correlated model,
as a function of the impropriety level $\lambda^2$.
As expected, the GLRT, which uses information from all sample impropriety coefficients, is here much more powerful than Roy's test,
especially in the high dimension case ($\gamma=2.5$) where $M$ and $N$ are close.

\subsubsection{Spiked model}
\label{sec:simuspike}
A Gaussian spike model with a single non-zero impropriety coefficient can be obtained
when the real and imaginary parts have a common contribution of rank one:
\begin{align}
\begin{split}
 {\bf u}_{m} &= {\bf s }_{m} + \sqrt{\theta} w_m \boldsymbol{\phi},\\
 {\bf v}_{m} &= {\bf t }_{m} + \sqrt{\theta} w_m \boldsymbol{\phi},
\end{split}
\label{eq:modspike}
\end{align}
where $\theta>0$, $\boldsymbol{\phi} \in \mathbb{R}^N$ is a normed deterministic vector {\color{blue} $||\boldsymbol{\phi}||_2=1$}, $w_m$ are i.i.d. Gaussian centered random
variables with unit variance,  and ${\bf s }_{m}$, ${\bf t }_{m}$ are Gaussian
i.i.d. vectors in $\mathbb{R}^N$, for $1 \le m \le M$.
This scenario depicts a case where a low-rank improper signal is corrupted by proper noise.
Straightforward computations
show that there is a single non-zero population impropriety coefficient, a {\em spike}, which expresses as
$\lambda_1= \frac{\theta} {1 + \theta}$.

Fig. \ref{fig:hist_spike} displays the empirical distribution of the squares of the sample impropriety coefficients
under alternative $H_1$ spiked model for different spike level $\lambda_1$. Again, the bulk of these
coefficients matches
very well the limiting distribution derived under the  properness hypothesis $H_0$, whatever the spike level.
In addition, for ``weak'' spikes, {\em i.e.} when $\lambda_1^2$ is small relative to the phase transition threshold $\rho_c$
defined in Theorem \ref{thm:spike},
the greatest sample impropriety coefficient $r_1$, which is an estimator of the spike power $\lambda_1^2$,
does not separate from this bulk and is stuck around the edge $c$ of the limiting
distribution. Conversely, for stronger spikes where $\lambda_1^2 > \rho_c$, $r_1$ clearly separates
from the bulk and concentrates around the limiting value $\overline{\rho}_1$.
This numerically supports Theorem \ref{thm:spike}.

\begin{figure}[htbp!]
    \begin{center}
    \includegraphics[width=0.99\linewidth]{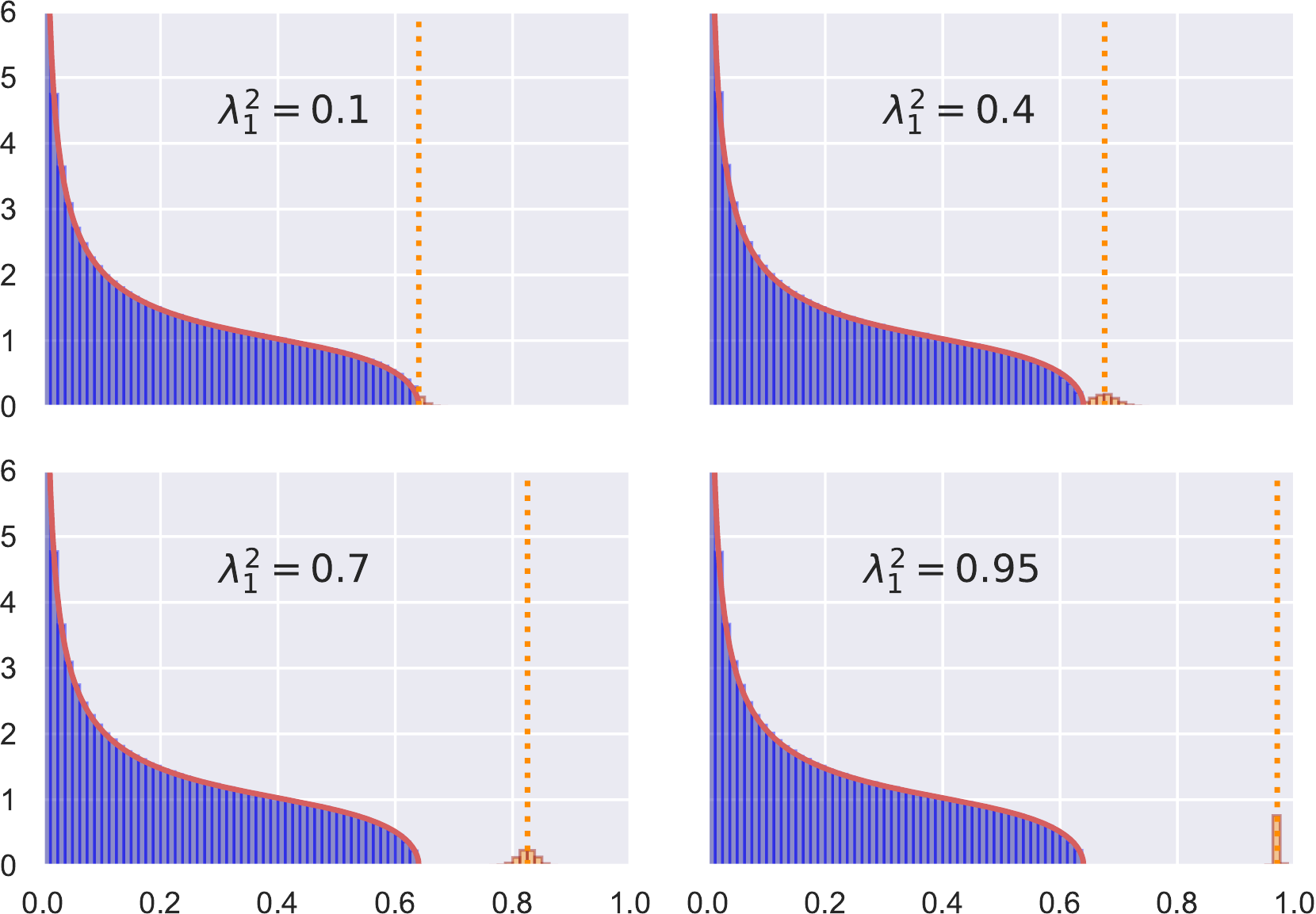}
    \end{center}
    \caption{Histogram of the squared sample impropriety coefficients under the
    spiked impropriety model described in \eqref{eq:modspike} with $N=100$
    and $\gamma=5$. Blue bars are for the bulk of the squared impropriety coefficients lower than the edge of the limiting distribution $c$, orange bars
    are for the ones greater than  $c$. The pdf of the limiting distribution under $H_0$ given in Theorem \ref{thm:pdfev} is shown in solid red  line.
    The limiting spike values   $\overline{\rho}_1$ defined in Theorem \ref{thm:spike} are depicted as vertical dotted orange  lines.
    The phase transition threshold is here $\rho_c=0.25$.
    Top row sub-figures are for a squared population impropriety coefficient
    $\lambda_1^2=0.1 < \rho_c$ (left) and $\lambda_1^2=0.4> \rho_c$ (right),
    bottom row is likewise for $\lambda_1^2=0.7> \rho_c$ and $\lambda_1^2=0.95>
\rho_c$ respectively. Histograms obtained with $1000$ Monte-Carlo runs.
    \label{fig:hist_spike}
    }
\end{figure}

In Fig. \ref{fig:power_spike} the power of both GLRT and Roy's tests are displayed as a function of the spike power $\lambda_1^2$.
This shows that for ``weak'' spikes, {\em i.e.} when $\lambda_1^2 < \rho_c$, the two tests have a very low power.
In fact, the largest sample impropriety coefficient $r_1$
does not separate from the bulk and cannot be detected correctly using Roy's test. It is interesting to note that
GLRT, which uses the information in all the coefficients, is here slightly ``more powerful'' to detect such weak spikes.
Nevertheless, as soon as $r_1$ separates from the bulk,  {\em i.e.} for stronger spikes where $\lambda_1^2 > \rho_c$, Roy's test
becomes much more powerful than the GLRT, with a power that converges quickly towards $1$ as expected.


\begin{figure}[htbp!]
    \begin{center}
    \includegraphics[width=0.8\columnwidth]{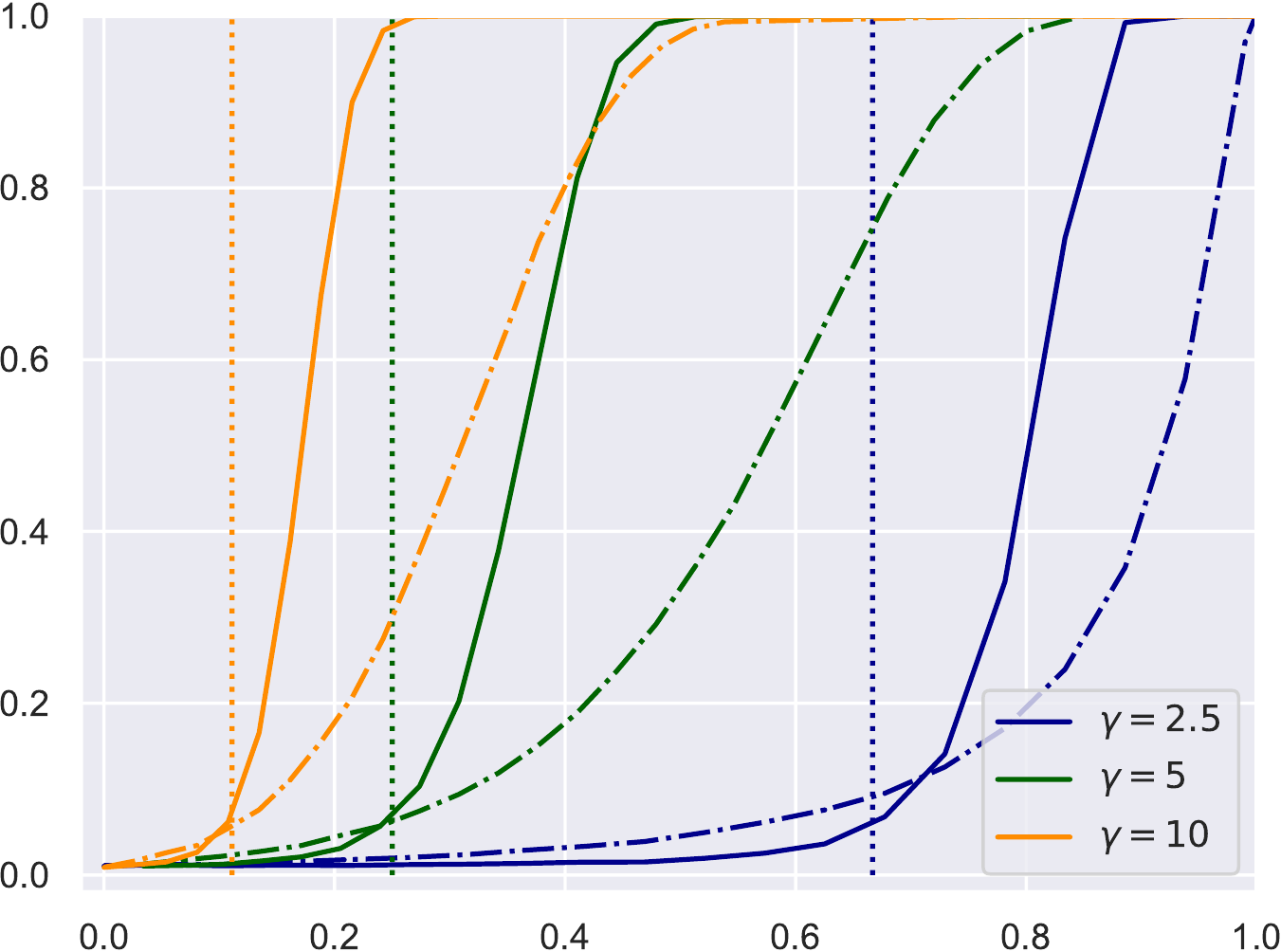}
    \end{center}
    \caption{Power of Roy's test, solid line, and GLRT, dashdotted line, vs
    the largest squared population impropriety coefficient $\lambda_1^2$  under the
    spiked impropriety model described in \eqref{eq:modspike}.
    The phase transition thresholds $\rho_c$
    given in Theorem \ref{thm:spike} are depicted as vertical dotted lines.
    Results are shown for different values of $\gamma$: blue lines for
$\gamma=2.5$, green for $\gamma=5$, orange for $\gamma=10$ (dimension $N=100$,
nominal value of false alarm probability $\alpha=0.01$).
Powers estimated with $1000$ Monte-Carlo runs.}
    \label{fig:power_spike}
\end{figure}

{\color{blue}
\subsubsection{Mixed scenario}
We consider now a mixed scenario with some highly improper, some less improper, and also some proper components. More precisely, we generate power imbalances between the real and imaginary parts as follows
\begin{align}
\begin{split}
 {\bf u}_{m} &= {\bf s }_{m}, \\
 {\bf v}_{m} &= {\bf t }_{m} + \sqrt{2\theta} {\bf w}_m,
\end{split}
\label{eq:modpca}
\end{align}
where $\theta>0$, ${\bf s }_{m}$, ${\bf t }_{m}$ are Gaussian
i.i.d. vectors in $\mathbb{R}^N$, for $1 \le m \le M$,
and  ${\bf w}_m$ are i.i.d. Gaussian centered vector
where
$p_k$ denotes, in a principal component analysis,  the fraction of variance explained by the $k$th principal component, hence $1 \ge p_1 \ge p_2 \ge  \ldots p_N \ge 0$ and $p_1 + \cdots +p_n=1$, and such that ${\bf w}_m$ have a unit total variance, {\em i.e.} $E\left[||{\bf w}_m||_2^2\right]=1$.
Thus the ordered population impropriety coefficients read $\lambda_k = \frac{\theta p_k}{1 + \theta p_k}$, for $1\le k \le N$.

Fig. \ref{fig:power_pca} shows the power of both GLRT and Roy's test as a function of the principal spike power $\lambda_1^2$.  In this setting, the dimension is $N=20$ and we have that:
\begin{itemize}
 \item $95\%$ of the variance of the source ${\bf w}_m$ is explained by its first five principal components ($p_1=50\%$, $p_2=20\%$, $p_3=p_4=10\%$, $p_5=5\%$) which yield
 the most significant (i.e. the largest) impropriety coefficients
 \item the remaining $5\%$ are explained by the following five components ($p_6= \ldots=p_{10}=1\%)$, which yield significantly smaller impropriety coefficients.
\end{itemize}
As a consequence, half of the impropriety coefficients are non zero,  while the other half are zero: $\lambda_{11}= \ldots = \lambda_{20}=0$.
This scenario mimics usual principal component analysis where the  interesting source lives on a lower dimensional space and its explained variance gradually decreases with the order of the principal components. These curves emphasize that GLRT, which uses all the sample coefficients, can be much more powerful than Roy's test for small sample size ($\gamma=2.5$ and $\gamma=5$). However for larger values of $\gamma$, Roy's test becomes significantly more powerful than GLRT for stronger spikes $\lambda_1^2$. Note also that even if half of the impropriety coefficients are non-zero,
the phase transition behavior given in Thm. \ref{thm:spike} is still present for Roy's test.

\begin{figure}[htbp!]
    \begin{center}
    \includegraphics[width=0.8\columnwidth]{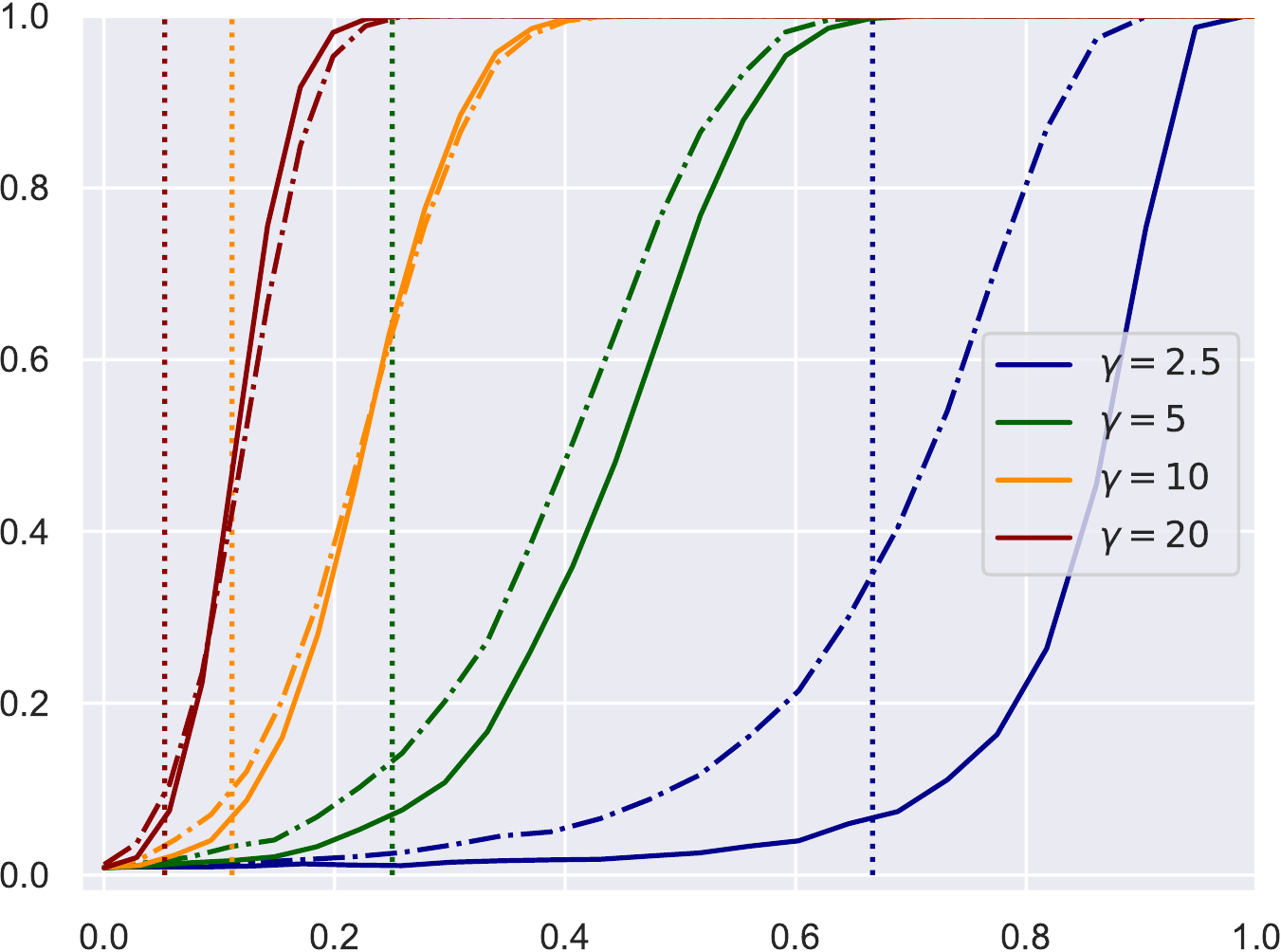}
    \end{center}
    \caption{\color{blue}Power of Roy's test, solid line, and GLRT, dashdotted line, vs
    the largest squared population impropriety coefficients $\lambda_1^2$  under the
    model described in \eqref{eq:modpca} for dimension $N=20$, with
    $p_1=50\%$, $p_2=20\%$, $p_3=p_4=10\%$, $p_5=5\%$, $p_6=\ldots=p_{10}=1\%$,
    and $p_k=0$ for $11 \le k \le 20$.
    The phase transition thresholds $\rho_c$
    given in Theorem \ref{thm:spike} are depicted as vertical dotted lines.
    Results are shown for different values of $\gamma$: blue lines for
$\gamma=2.5$, green for $\gamma=5$, orange for $\gamma=10$  and red for
 $\gamma = 20$ (nominal value of false alarm probability $\alpha=0.01$).
Powers estimated with $1000$ Monte-Carlo runs.}
    \label{fig:power_pca}
\end{figure}

}

\section{Relation with existing work on CCA}
\label{sec:discuss}
In CCA, one is classically interested in testing for independence between $N$-dimensional  real or complex Gaussian random vectors ${\bf x}$ and ${\bf y}$. Invariant statistics to reject the null hypothesis of independence consist in
the squared canonical correlation coefficients.
%
When  ${\bf x_i}$ and ${\bf y_i}$, $1 \le i \le M$, are $M$ i.i.d. copies of $N$-dimensional real Gaussian independent vectors, the $N$ squared sample canonical correlation coefficients are known to be jointly distributed as the eigenvalues of a matrix-variate beta distribution
$\mathcal{B}_N(\frac{n_1}{2},\frac{n_2}{2} )$,
with $n_1=N$ and $n_2=M-N$ as shown in \cite[Section 11.3]{Muirhead:2005}.

For the impropriety case, Prop. \ref{prop:pdfBeta} shows that the $N$ squared impropriety coefficients are jointly distributed as the eigenvalues of a matrix-variate beta distribution $\mathcal{B}_N(\tfrac{n_1}{2},\tfrac{n_2}{2})$
with parameters $n_1=N+1$ and $n_2=M-N$. Surprisingly,
this is quite close to the real and independent CCA case, the difference being a $+1$ offset in the $n_1$ parameter. Despite this similarity, it is important to note the following points.

\subsection*{Independence is not compatible with impropriety testing}
The two regimes of parameters {\color{blue}(in the matrix-variate beta distributions)} obtained for the real and independent CCA case or the impropriety case remain incompatible. This is intrinsically due to the fact that the underlying assumptions are not compatible. As shown for instance by \eqref{eq:CCA}, the impropriety coefficients are the canonical correlation coefficients between vectors ${\bf z}$ and ${\bf z^*}$ which are complex-valued and fully dependent.

\subsection*{Impropriety testing is a more structured problem}

Consequently we do not think that classical results for CCA between real, or complex, independent vectors can be readily extended or adapted to the impropriety detection problem. Our contribution is thus a way to overcome this and to provide new insights on the impropriety problem. In this paper we have precisely proposed a direct characterization of the usual impropriety testing statistics.
Furthermore, to our knowledge, the numerous existing works on impropriety have never been able to easily adapt these CCA results to characterize impropriety coefficients in finite or even asymptotic regimes.


\subsection*{The $+1$ offset is non-negligible}
In the high dimensional regime where both $M$ and $N$ go to infinity, it might be tempting to consider that the $+1$ offset in the $n_1$ parameter of the matrix-beta distribution becomes negligible, and thus that the distributions of the  statistics for testing the independence between real Gaussian vectors or for the impropriety detection problem are asymptotically equivalent. This is actually incorrect.  The $+1$ offset modifies the asymptotic distribution of any linear spectral statistic such as the GLRT, but also of the largest eigenvalue one, and therefore of Roy's test. For instance,
 Theorem 6 and its demonstration allows us to show that the asymptotic mean $m'$ of the GLRT statistics derived for the parameters $n_1=N$ and $n_2=M-N$ (real and independent CCA case) is $m'= m + \ln{\frac{\gamma-1}{\gamma}}$ where $m$ is the asymptotic mean of the impropriety GLRT statistics given in Theorem 6, while the asymptotic variance remains bounded. An illustration of this mismatch for both GLRT and Roy's test is depicted in Fig. \ref{fig:CCAcomp}.

\begin{figure}[hbtp!]
  \centering
  \includegraphics[width=\columnwidth]{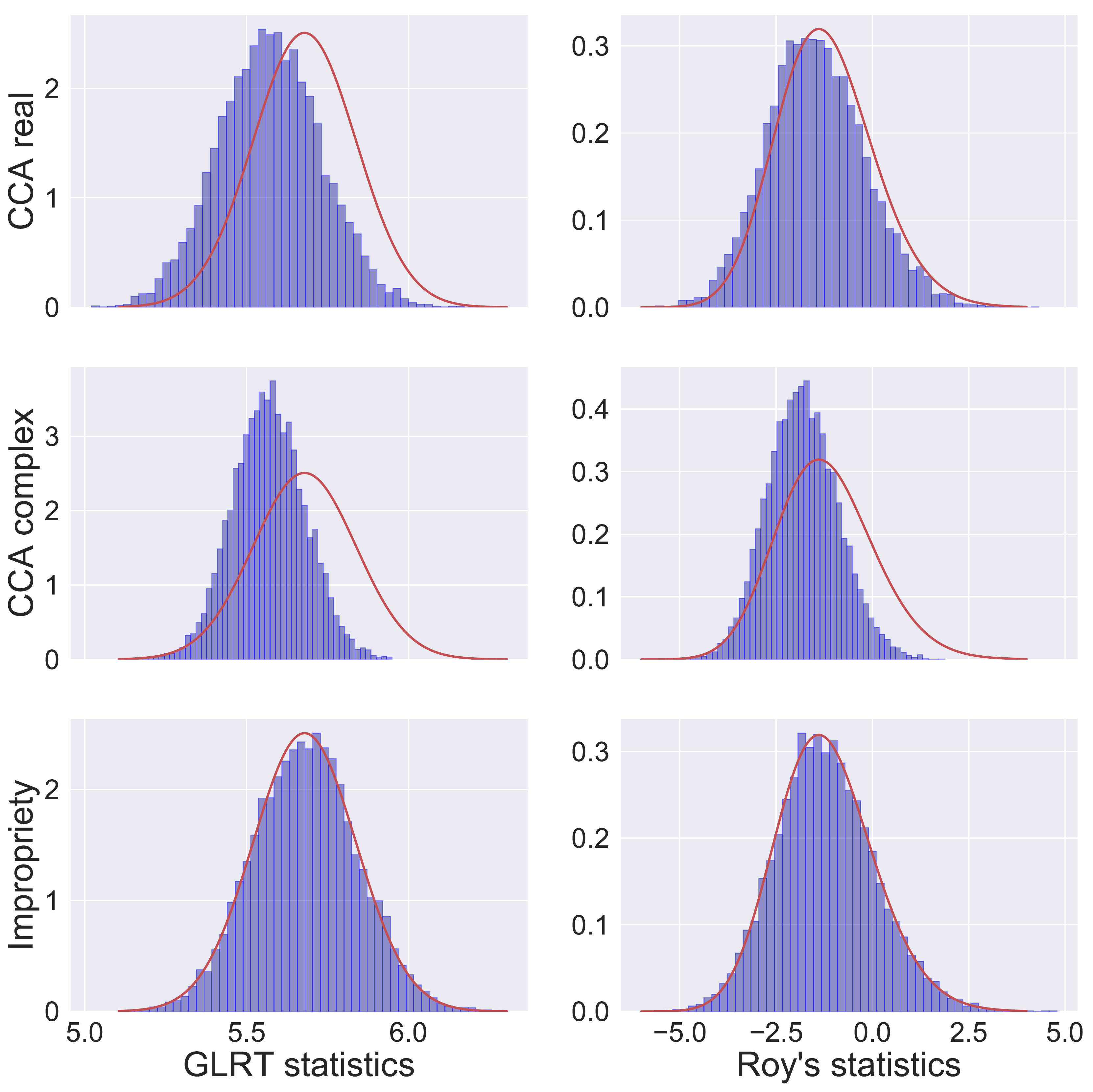}
  \caption{Histograms for $M=500$ and $N = 50$ ($\gamma=10$) of the GLRT statistics (left column) and Roy's statistics (right column). The pdf of the GLRT and Roy's test limiting distribution for impropriety testing given in Thm. \ref{thm:asympt} and Thm. \ref{thm:Roy} are shown in solid red line in the left and right column respectively.
  Top row:  CCA statistics for {independent} and {real} vectors.
  Middle row:  CCA for {independent} and {complex} vectors.
  Bottom row: Impropriety statistics. Histograms obtained with 10000 Monte-Carlo runs.\label{fig:CCAcomp}}
\end{figure}

\section{Concluding remarks}
Properness testing for complex Gaussian random vectors in the asymptotic regime
 relies on the characterization of sample impropriety coefficients. In particular,
their limiting distributions give access to the behavior of classical
 GLRT and Roy's test. The results presented in this article
demonstrate that the asymptotic regime is actually reached quite rapidly in
practice and the proposed original  approximations are well-suited to
a wide range of complex-valued datasets.

The phase transition highlighted in Roy's test has also potential
applications in the search for complex-valued low-rank signals corrupted by
proper noise in large datasets. {\color{blue}In this context, the proposed high dimensional approximations can be extended to the sequential testing problem \cite{Novey:2011} to estimate the number of improper sources, {\em i.e.} of non-zero impropriety coefficients.}
Another natural extension of the
proposed work consists of considering the case of
quaternion random vectors which possess several properness levels, thus trying
to decipher their correlation symmetry patterns. This could be
helpful in the spectral characterization of bivariate signals
\cite{flamant2017} among other quaternion signal processing applications.



\section*{Acknowledgment}
The authors would like to thank Prof. Romain Couillet for his many valuable comments and suggestions.

\appendices

\section{Proof of Theorem \ref{thm:pdfev} (limiting empirical distribution)}

\label{app:pdf}
{\color{blue}
According to Proposition \ref{prop:pdfBeta}, the sample squared  impropriety
coefficients are distributed under $H_0$ as the eigenvalues  of a matrix-variate Beta distribution $\mathcal{B}_N(\frac{1}{2} n_1,\frac{1}{2}  n_2)$. Moreover, \cite[Thm 3.3.1, p. 109]{Muirhead:2005} shows that the eigenvalues of  a  $\mathcal{B}_N(\frac{1}{2} n_1,\frac{1}{2}  n_2)$ random matrix are distributed as the eigenvalues of ${\bf A}({\bf A}+{\bf B})^{-1}$, where
${\bf A}, {\bf B} \in {\mathbb R}^{N \times N}$ are two independent random matrices  with respective distributions the Wishart laws $W_N(n_1,{\bf I}_N)$ and $W_N(n_2,{\bf I}_N)$, where $n_1$ and $n_2$ are their respective number of {\em degrees of freedom} and where the Wishart {\em scaling matrix} parameter  is set to the identity matrix ${\bf I}_N$. Thus the limiting distribution for
the impropriety coefficients can be derived as the limiting distribution of the  eigenvalues of ${\bf A}({\bf A}+{\bf B})^{-1}$.

Assume now that we are in the asymptotic regime where $N,n_1,n_2\, \rightarrow \, +\infty$ with $N/n_1 \, \rightarrow \, d \in (0,1]$ and
$N/n_2 \, \rightarrow \, d' \in (0,1)$.

Then, as demonstrated in \cite{silverstein1985}, the empirical law of the eigenvalues of the  following $F$-matrix
$\tfrac{d}{d'} {\bf A} {\bf B}^{-1}$ converges to a distribution with density given as:
\begin{align}
f(x) = \frac{ (1-d')  \sqrt{ (x-a)(b-x) } } { 2 \pi  x  (x d + d') },
\label{eq:pdfproof}
\end{align}
on the interval $x \in [a,b]$, where
$a= \left( \tfrac{ 1 - \sqrt{ 1 -(1-d) (1-d') } } { 1 - d' } \right)^2$ and
$b= \left( \tfrac{ 1 + \sqrt{ 1 -(1-d) (1-d') } } { 1 - d' } \right)^2$.

Note that each eigenvalue $r_i$ of ${\bf A}({\bf A}+{\bf B})^{-1}$ can be deduced from each eigenvalue $x_i$ of $\tfrac{d}{d'} {\bf A} {\bf B}^{-1}$ thanks to the relation $r_i= \tfrac{d' x_i}{d+ d'x_i}$.
The {continuous mapping theorem} ensures therefore that the asymptotic law of the eigenvalues of a $\mathcal{B}_N(\frac{1}{2} n_1,\frac{1}{2}  n_2)$  random matrix  can be directly deduced from \eqref{eq:pdfproof} using the aforementioned change of variable.

Last, according to proposition \ref{prop:pdfBeta}, the parameters for the matrix-variate Beta distribution in our case are $n_1= N+1$ and $n_2=M-N$. Due to the asymptotic regime stated above, one gets
that $ d= \lim\,  N/n_1  = 1 $ and $d'= \lim\, N/n_2= \tfrac{1}{\gamma-1}$. Plugging this parameter values in the asymptotic distribution of the eigenvalues of the  $\mathcal{B}_N(\frac{1}{2} n_1,\frac{1}{2}  n_2)$ random matrix deduced from \eqref{eq:pdfproof} gives finally the limiting density given in Theorem \ref{thm:pdfev} and concludes the proof.
}

\section{Proof of Theorem \ref{thm:asympt} (central limit theorem)}

\label{app:asympt}

According to theorem \ref{thm:GLRT}, $T'= \sum_{n=1}^N \zeta_n$ where the $\zeta_n$ are independent random variables such that $\zeta_n= - \ln u_n$ with $u_n$ beta distributed s.t.
$u_n \sim \mathcal{B}\left( a_n, b \right)$, $a_n=\tfrac{M-N-n+1}{2}$, $b=\tfrac{N+1}{2}$,  for $1\le n \le N$.
Based on the centered moments of a logarithmically transformed beta-distributed variable as given in \cite{nadarajah2006},
then $E [\zeta_n]= \psi(a_n+b)- \psi(a_n)$ where $\psi(\cdot)$ is the digamma function, and
$\var[\zeta_n]= \psi_1(a_n) - \psi_1(a_n+b)$ where $\psi_1(\cdot)$ is the trigamma function.
Using asymptotic expansions of the digamma and trigamma functions for large argument $x$,
\begin{align*}
\psi(x)& = \log{x} - \tfrac{1}{2x} + O\left(\tfrac{1}{x^2}\right),
\quad \psi_1(x)=  \tfrac{1}{x} +\tfrac{1}{2x^2} + O\left(\tfrac{1}{x^3}\right),
\end{align*}
  plus Stirling formula $\log{(n!)}= n \log{n} -n + \tfrac{1}{2}
 \log{(2\pi n)} + O\left(\tfrac{1}{n}\right)$ for large $n$,
straightforward computations omitted here for the sake of brevity yield that
$E[T']= \sum_{n=1}^N E[\zeta_n]= m_M + O(1/M)$ and $\var(T')=\sum_{n=1}^N \var(\zeta_n) = s_M^2 + O(1/M^2)$.

In order to apply Lyapunov central limit theorem \cite[p. 362]{billingsley1995probability} to
${T'= \sum_{n=1}^N \zeta_n}$, it is sufficient to show that
\begin{align*}
 \frac {1}{\var(T')^2} \sum_{n=1}^N E\left[ \left( \zeta_n  - E[\zeta_n] \right)^4 \right]\, \rightarrow \, 0.
\end{align*}
The expression of the fourth order centered moment of $\zeta_n$ gives that
$E\left[ \left( \zeta_n  - E[\zeta_n] \right)^4 \right] = O\left( 1/(M-n+2)^2 \right)$ for $1\le n \le N$. Then
$\sum_{n=1}^N E\left[ \left( \zeta_n  - E[\zeta_n] \right)^4 \right]= O(1/M)$.
As $\var(T')= s_{M}^2 + O(1/M^2)$, it comes that $\var(T')^{-2}= O(1)$ and
 the previous Lyapunov sufficient condition holds. Thus
$$Z \equiv \frac {1}{\sqrt{ \var(T') }} \sum_{n=1}^N \left( \zeta_n  - E[\zeta_n]  \right) \ {\xrightarrow {d}} \ \mathcal{N}(0,1).$$
By noting finally that $\frac{1}{s_M} (T'-m_M)= Z + O(1/M)$, Slutsky's theorem allows us to conclude the proof.

\bibliography{references.bib}
\bibliographystyle{IEEEtran}

\end{document}